\documentclass{article}
\usepackage{graphicx, braket, amsmath, amssymb, amsthm}
\usepackage[colorlinks = true]{hyperref}

\newtheorem{definition}{Definition}
\newtheorem{theorem}{Theorem}
\newtheorem{lemma}{Lemma}
\newtheorem{corollary}{Corollary}

\DeclareMathOperator{\trace}{trace}
 
\begin{document}

\title{Sharma-Mittal Quantum Discord}
\author{Souma Mazumdar\thanks{Email: \texttt{souma.mazumdar@bose.res.in}}, Supriyo Dutta\thanks{Email: \texttt{supriyo.dutta@bose.res.in}} , Partha Guha \thanks{Email: \texttt{partha@bose.res.in}} \\ Depertment of Theoretical Sciences \\ S. N. Bose National Centre for Basic Sciences\\ Block - JD, Sector - III, Salt Lake City, Kolkata - 700 106}
\date{}

\maketitle

\begin{abstract}
We demonstrate a generalization of quantum discord using a generalized definition of von-Neumann entropy, which is Sharma-Mittal entropy; and the new definition of discord is called Sharma-Mittal quantum discord. Its analytic expressions are worked out for two qubit quantum states as well as Werner, isotropic, and pointer states as special cases. The R{\'e}nyi, Tsallis, and von-Neumann entropy based quantum discords can be expressed as limiting case for of Sharma-Mittal quantum discord. We also numerically compare all these discords and entanglement negativity.
\end{abstract}

\section{Introduction}

Entropy is a measure of information generated by random variables. In classical information theory \cite{cover2012elements}, we use Shannon entropy as a measure of information. In quantum information theory, the von-Neumann entropy is the generalization of classical Shannon entropy. In resent years, a number of entropies has been proposed to measure information, for instance the R{\'e}nyi entropy \cite{renyi1961measures}, and Tsallis entropy \cite{tsallis1988possible}, in different contexts. The R{\'e}nyi entropy is a generalization of the Hartley entropy, the Shannon entropy, the collision entropy and the min-entropy. Also, the Tsallis entropy is a generalization of Boltzmann-Gibbs entropy. The Sharma-Mittal entropy generalizes both R{\'e}nyi entropy and Tsallis entropy. Further generalizations of entropy functions are also available in the literature \cite{tempesta2015theorem}, \cite{tempesta2016beyond}.

In quantum information and computation \cite{mcmahon2007quantum}, quantum discord is a well known quantum correlation, which was first introduced in \cite{ollivier2001quantum}, \cite{henderson2001classical} as a quantum mechanical difference between two classically equivalant definitions of mutual information. In classical information theory, the Shannon entropy of a random varible $X$ is defined by $H(X) = - \sum_{X = x}P(X = x)\log(P(X = x))$, where $P$ is the probability function. Given a joint random variable $(X, Y)$, the mutual information $\mathcal{I}(X; Y)$ can be expressed by the following equations: 
\begin{align}
\mathcal{I}(X; Y) & = H(X) + H(Y) - H(X, Y) \label{equation_1}\\
& = H(Y) - H(Y|X) \label{equation_2}.
\end{align}
In quantum information theory, a quantum state is represented by a density matrix $\rho$ which is equivalent to the random variable $X$. Also, the Shannon entropy was generalised by the von-Neumann entropy, given by $H(\rho) = - \sum_i \lambda_i \log(\lambda_i)$, where $\lambda_i$ is an eigenvalue of $\rho$. If $\rho$ represents a bipartite quantum state in $\mathcal{H}^{(a)} \otimes \mathcal{H}^{(b)}$, that is a state distributed between two parties $A$ and $B$, then the reduced density matrices $\rho_a$ and $\rho_b$ are considered as the substitutes of the marginal probability distributions. Then, equation (\ref{equation_1}) can be generalised as quantum mutual information $\mathcal{I}(\rho)$ which is expressed as,
\begin{equation}
\mathcal{I}(\rho) = H(\rho_a) + H(\rho_b) - H(\rho).
\end{equation}
Also, the equation (\ref{equation_2}) can be generalised as
\begin{equation}
\mathcal{S}(\rho) = H(\rho_a) - \max_{\{\Pi_i\}}\left[\sum_i p_i H(\rho^{(i)})\right],
\end{equation}
where the maximization runs over all possible projective measurements $\{\Pi_i\}$ on $\mathcal{H}^{(b)}$, as well as 
\begin{equation}
\rho^{(i)} = (I_a \otimes \Pi_i)^\dagger \rho (I_a \otimes \Pi_i), ~\text{and}~ p_i = \trace(\rho^{(i)}). 
\end{equation}
The quantum discord is defined by $\mathcal{D}(\rho) = \mathcal{I}(\rho) - \mathcal{S}(\rho)$ \cite{ollivier2001quantum}. An analytic expression of quantum discord based on von-Neumann entropy for two qubit states is constructed in \cite{luo2008quantum}. In resent years, the von-Neumann entropy was replaced by the R{\'e}nyi entropy, and the Tsallis entropy \cite{hou2014quantum}, \cite{seshadreesan2015renyi}, \cite{jurkowski2013quantum}. A recent review on quantum discord and its allies is \cite{bera2017quantum}.

In this article, we replace the von-Neumann entropy with the Sharma-Mittal entropy in the definition of quantum discord, first time in literature. Throughout the article, we consider logarithm with respect to the base $2$. The Sharma-Mittal entropy \cite{sharma1975entropy}, \cite{mittal1975some} of a random variable $X$ is denoted by, $H_{q, r}(X = x)$, and defined by,
\begin{equation}\label{Sharma_Mittal_entropy}
H_{q, r}(X = x) = \frac{1}{1 - r}\left[\left(\sum_{X = x}(P(X = x))^q\right)^{\frac{1 - r}{1 - q}} - 1\right],
\end{equation}
where $q$ and $r$ are two real parameters $q > 0, q \neq 1$, and $r \neq 1$ \cite{akturk2007sharma}, \cite{nielsen2011closed}. If limit $r \rightarrow 1$ in the above expression, we get the R{\'e}nyi entropy:
\begin{equation}\label{Renyi_entropy}
H^{(R)}_q(X) = \lim_{r \rightarrow 1} H_{q, r}(X = x) = \frac {1}{1 - q}\log \left(\sum_{X = x} (P(X = x))^q\right),
\end{equation} 
where $q \geq 0$ and $q \neq 1$. Similarly, if $r \rightarrow q$ in the equation (\ref{Sharma_Mittal_entropy}) we get the Tsallis entropy:
\begin{equation}\label{Tsallis_entropy}
H^{(T)}_{q}(X) = \lim_{r \rightarrow q} H_{q, r}(X = x) = \frac{1}{1-q}\left(\sum_{X = x}(P(X = x))^q - 1\right), 
\end{equation}
where $q \geq 0$ and $q \neq 1$. The Sharma-Mittal entropy reduced to Shannon entropy when both $q \rightarrow 1$ and $r \rightarrow 1$, which is
\begin{equation}\label{Shannon_entropy}
\begin{split} 
H(X) = & \lim_{(q, r) \rightarrow (1,1)}H_{q, r}(X = x) = \lim_{q \rightarrow 1}\lim_{r \rightarrow q}H_{q, r}(X = x) = \lim_{q \rightarrow 1}H^{(T)}_{q}(X) \\  
= & -\sum_{x \in X} P(X = x)\log(P(X = x)).
\end{split} 
\end{equation}
Recall that, a density matrix is a positive semi-definite, Hermitian matrix with unit trace. Hence, all its eigenvalues are non-negative. 
\begin{definition}
	Given a density matrix $\rho$ as well as two real numbers $q$ and $r$ with $q > 0, q \neq 1, r \neq 1$, the Sharma-Mittal entropy of $\rho$ is defined by,
	\begin{equation}
	H_{q, r}(\rho) = \frac{1}{1 - r}\left[\left(\sum_i(\lambda_i)^q\right)^{\frac{1 - r}{1 - q}} - 1\right],
	\end{equation}
	where $\lambda_i$s are eigenvalues of $\rho$. 
\end{definition}
The quantum mechanical counterparts of R{\'e}nyi and Tsallis entropy can also be generalised in a similar fashion. Also, the von-Neumann entropy is the alternative of Shannon entropy in quantum mechanical context.

\begin{definition}\label{sm_discord_deff}
	The Sharma-Mittal quantum discord of a quantum state in $\mathcal{H}^{(a)} \otimes \mathcal{H}^{(b)}$, represented by a density matrix $\rho$ is defined by $|\mathcal{D}_{q,r}(\rho)|$, where
	$$\mathcal{D}_{q,r}(\rho) = H_{q,r}(\rho_b) + \max_{\{\Pi_i\}}\left[\sum_i p_i H_{q,r}(\rho_a^{(i)})\right] -  H_{q,r}(\rho).$$ 
\end{definition}
In the above definition, replacing $H_{q,r}$ by $H^{(R)}_{q}$, $H^{(T)}_{q}$ and $H$, we get the R{\'e}nyi, Tsallis, and von-Neumann quantum discord, which are denoted by $|\mathcal{D}^{(R)}_q|$, $|\mathcal{D}^{(T)}_q|$, and $\mathcal{D}$, respectively. It is proved that for all density matrices $\rho$ the von-Neumann discord $\mathcal{D}(\rho) \geq 0$ \cite{datta2010condition}, which may not be true for  $\mathcal{D}_{q,r}(\rho), \mathcal{D}^{(R)}_q(\rho)$, and $\mathcal{D}^{(T)}_q(\rho)$. For instance, consider a Werner state given by $\rho = \begin{bmatrix}.2 & 0 & 0 & 0\\ 0 & .3 & -.1 & 0 \\ 0 & -.1 & .3 & 0 \\ 0 & 0 & 0 & .2 \end{bmatrix}$. We can calculate $\mathcal{D}_{.5, .4}(\rho) = -.3992$, and $\mathcal{D}^{(T)}_{.5}(\rho) = -.3057$. The idea of negative correlation has been considered in the theory of probability, but not well accepted in the quantum information theoretic community, till date.  Therefore, to assure non-negativity we consider the absolute value in the definition \ref{sm_discord_deff}.

The motivation behind this work is to generalize the idea of quantum discord in an unified fashion. Calculating an analytical expression of discord is a challenging task, as it involves an optimization term \cite{huang2014computing}. In the next section, we construct the analytic expression of $\mathcal{D}_{q,r}(\rho)$  for two qubit states, as well as $\mathcal{D}^{(R)}_q$, $\mathcal{D}^{(T)}_q$, and $\mathcal{D}$ as its limiting cases. Also, we compare each of them with entanglement negativity, which is a well known measure of entanglement \cite{horodecki2009quantum}. We also calculate the Sharma-Mittal quantum discord and its specifications for Werner, isotropic, and pointer states. Then, we conclude this article with some open problems.

\section{Discord for two qubit states}

\subsection{Discord in general}

Two qubit states which belongs to the Hilbert space $\mathcal{H}^2 \otimes \mathcal{H}^2$, act as the primary building blocks for encoding correlations in quantum information theory. The computational basis of $\mathcal{H}^2 \otimes \mathcal{H}^2$ is $\{\ket{00}, \ket{01} \ket{10}, \ket{11}\}$. Recall that, the Pauli matrices are
\begin{equation}
 \sigma_1 = \begin{bmatrix} 0 & 1 \\ 1 & 0 \end{bmatrix} , \sigma_2 = \begin{bmatrix} 0 & i \\ -i & 0 \end{bmatrix}, ~\text{and}~ \sigma_3 = \begin{bmatrix} 1 & 0 \\ 0 & -1 \end{bmatrix}.
\end{equation}
The identity matrix of order $2$ is given by $I_2 = \begin{bmatrix} 1 & 0  \\ 0 & 1 \end{bmatrix}$.
In general, any two qubit state is locally unitary equivalent to:
\begin{equation}
\gamma = \frac{1}{4}\left(I + \overrightarrow{a} \overrightarrow{\sigma} \otimes I + I \otimes \overrightarrow{b}\overrightarrow{\sigma} + \sum_{j = 1}^3c_j \sigma_j \otimes \sigma_j \right).
\end{equation} 
For analytical simplicity, and since we are just keen on the correlations in the bipartite states, we consider those states with the maximally mixed marginals, that is, we consider the following states:
\begin{equation}\label{density_matrix}
\rho = \frac{1}{4}\left(I + c_1\sigma_1 \otimes \sigma_1 + c_2\sigma_2 \otimes \sigma_2 + c_3\sigma_3 \otimes \sigma_3 \right),
\end{equation} 
where $c_1, c_2$, and $c_3$ are real numbers. Expanding the tensor products we get,
\begin{equation}\label{density_matrix_expanded}
\rho = \frac{1}{4}\begin{bmatrix} 1+c_{3} & 0 & 0 & c_{1} - c_{2} \\ 0 & 1 - c_{3} & c_{1}+c_{2} & 0 \\ 0 & c_{1}+c_{2} & 1 - c_{3} & 0 \\ c_{1} - c_{2} & 0 & 0 & 1+c_{3} \end{bmatrix}.
\end{equation}

\begin{lemma}\label{values_of_c}
	If a matrix $\rho = \frac{1}{4}\left(I + c_1\sigma_1 \otimes \sigma_1 + c_2\sigma_2 \otimes \sigma_2 + c_3\sigma_3 \otimes \sigma_3 \right),$ represents a density matrix of a quantum state, then $-1 \leq c_i \le 1$ for $i = 1, 2, 3$, and $c_1 + c_2 + c_3 \leq 1$.
\end{lemma}

\begin{proof}
	The eigenvalues of $\rho$ in terms of $c_1, c_2$, and $c_3$ are,
	\begin{equation}\label{eigenvalues_in_gen}
	\begin{split}
	\lambda_0 = & \frac{1}{4}(1 - c_1 - c_2 - c_3), \lambda_1 = \frac{1}{4}(1 - c_1 + c_2 + c_3), \\
	\lambda_2 = & \frac{1}{4}(1 + c_1 - c_2 + c_3), ~\text{and}~ \lambda_3 = \frac{1}{4}(1 + c_1 + c_2 - c_3). \\   
	\end{split}
	\end{equation}
	As $\rho$ is a Hermitian matrix, all the eigenvalues are real. Also, $\rho$ is positive semi-definite matrix. Therefore, $\lambda_i \ge 0$. As $\trace(\rho) = 1$, we have $\lambda_0 + \lambda_1 + \lambda_3 + \lambda_4 = 1$. As $\lambda_0 \geq 0$ we have $c_1 + c_2 + c_3 \leq 1$. In addition, $\lambda_0 + \lambda_1 \geq 0$ indicate $c_1 \leq 1$ and $\lambda_2 + \lambda_3 \geq 0$ indicates $c_1 \geq -1$. Combining we get $-1 \leq c_1 \leq 1$. Similarly, $-1 \leq c_2 \leq 1$, and $-1 \leq c_3 \leq 1$.
\end{proof}

Eigenvalues of a density matrix $\rho$ suggest that the Sharma-Mittal entropy of $\rho$ is  
\begin{equation}\label{sharma_mittal_for_2_qubits}
\begin{split}
H_{q,r}(\rho) = \frac{1}{1 - r}[\frac{1}{4^{\frac{q(1 - r)}{(1 - q)}}}\{(1 - c_1 - c_2 - c_3)^q  + (1 - c_1 + c_2 + c_3)^q & \\
+ (1 + c_1 - c_2 + c_3)^q + (1 + c_1 + c_2 - c_3)^q \}^{\frac{1 - r}{1 - q}}-1].&
\end{split}
\end{equation}
The R{\'e}nyi entropy of $\rho$ will be obtained by taking $\lim_{r \rightarrow 1}H_{q,r}(\rho)$ which is
\begin{equation}\label{renyi_for_2_qubits}
\begin{split} 
H^{(R)}_q(\rho) = & \frac{1}{1 - q}\log(\frac{1}{4^q}\{(1 - c_1 - c_2 - c_3)^q \\ 
& + (1 - c_1 + c_2 + c_3)^q + (1 + c_1 - c_2 + c_3)^q\\
& + (1 + c_1 + c_2 - c_3)^q\}).
\end{split} 
\end{equation}
Similarly, the Tsallis entropy will be $\lim_{r \rightarrow q}H_{q,r}(\rho)$ which is
\begin{equation}\label{tsallis_for_2_qubits}
\begin{split}
H^{(T)}_q(\rho) = \frac{1}{1 - q}[\frac{1}{4^q}\{(1 - c_1 - c_2 - c_3)^q + (1 - c_1 + c_2 + c_3)^q & \\
+ (1 + c_1 - c_2 + c_3)^q + (1 + c_1 + c_2 - c_3)^q \} - 1].
\end{split}
\end{equation}
Also, the von-Neumann entropy of $\rho$ is
\begin{equation}
\begin{split} 
H(\rho) = & -\frac{(1 - c_1 - c_2 - c_3)}{4}\log\frac{(1 - c_1 - c_2 - c_3)}{4}\\
& - \frac{(1 - c_1 + c_2 + c_3)}{4}\log\frac{(1 - c_1 + c_2 + c_3)}{4} \\
& - \frac{(1 + c_1 - c_2 + c_3)}{4}\log\frac{(1 + c_1 - c_2 + c_3)}{4} \\
& + \frac{(1 + c_1 + c_2 - c_3)}{4}\log\frac{(1 + c_1 + c_2 - c_3)}{4}.
\end{split} 
\end{equation}			
\begin{theorem}\label{sm_based_discord_in_gen}
	The Sharma-Mittal quantum discord of quantum state represented by a density matrix $\rho$, mentioned in equation (\ref{density_matrix}), is 
	\begin{equation*}
	\mathcal{D}_{q, r}(\rho) =  \frac{1}{(1 - r)}\left[\left\{\frac{(1 + c)^q}{2^q} + \frac{(1 - c)^q}{2^q} \right\}^{\frac{1 - r}{1 - q}} - 1\right] + \frac{2^{1 - r} - 1}{1 - r} - H_{q, r}(\rho),
	\end{equation*}
	where $c = \max\{|c_1|, |c_2|, |c_3|\}$, as well as $q$ and $r$ are two real numbers, such that $q > 0, q \neq 1, r \neq 1$.
\end{theorem}

\begin{proof} 
	Recall the definition \ref{sm_discord_deff} of Sharma-Mittal quantum discord. We can check that the reduced density matrices over the subsystems $\mathcal{H}^{(a)}$ and $\mathcal{H}^{(b)}$ are given by $\rho_a = \rho_b = \frac{I_2}{2}$, which is a density matrix which eigenvalues $\frac{1}{2}$. Now, the Sharma-Mittal entropy of $\frac{I}{2}$ is
	\begin{equation}\label{entropy_i_by_2}
	H_{q, r}(\rho_a) = H_{q, r}(\rho_b) = H_{q, r}\left(\frac{I}{2}\right) = \frac{2^{1 - r} - 1}{1 - r}.
	\end{equation}
	The local measurements for the party $B$ along the computational basis $\{\ket{k}\}$ is $\{\ket{k}\bra{k}: k = 0, 1\}$. Now, any von-Neumann measurement for the party $B$ is given by 
	\begin{equation}
	B_k = V\ket{k}\bra{k}V^\dagger: k = 0,  1, V \in U(2).
	\end{equation}
	Any unitary operator $V \in U(2)$ can be expressed as 
	\begin{equation}
	V = t I + iy_1 \sigma_1 + i y_2 \sigma_2 + i y_3 \sigma_3,
	\end{equation}
	where $t, y_1, y_2$, and $y_3 \in \mathbb{R}$ and $t^2 + y_1^2 + y_2^2 + y_3^2 = 1$. After measurement the state $\rho$ will be changed to an ensemble $\{p_k, \rho^{(k)}\}$, where 
	\begin{equation}
	\begin{split} 
	\rho^{(k)} = & \frac{1}{p_k}(I_2 \otimes B_k) \rho (I_2 \otimes B_k) \\
	= & \frac{1}{p_k}(I_2 \otimes V\ket{k}\bra{k}V^\dagger) \rho (I_2 \otimes V\ket{k}\bra{k}V^\dagger) \\
	= & \frac{1}{p_k}(I_2 \otimes V)(I_2 \otimes \ket{k}\bra{k})(I_2 \otimes V)^\dagger \rho (I_2 \otimes V)(I_2 \otimes \ket{k}\bra{k}) (I_2 \otimes V)^\dagger .
	\end{split} 
	\end{equation} 
	Simplifying we get,
	\begin{equation}
	\begin{split}
	\rho^{(0)} = & \frac{1}{2}(I + c_1z_1\sigma_1 + c_2 z_2 \sigma_2 + c_3 z_3 \sigma_3)\otimes (V \ket{0}\bra{0}V^\dagger),  \\ 
	\text{and}~\rho^{(1)} = & \frac{1}{2}(I - c_1z_1\sigma_1 - c_2 z_2 \sigma_2 - c_3 z_3 \sigma_3)\otimes (V \ket{1}\bra{1}V^\dagger),
	\end{split}
	\end{equation} 
	where $z_1 = 2(-ty_2 + y_1y_3), z_2 = 2(ty_1 + y_2y_3)$, and $z_3 = (t^2 + y_3^2 - y_1^2 - y_2^2)$. Also, $p_1 = p_2 = \frac{1}{2}$. It can be verified that $z_1^2 + z_2^2 + z_3^2 = 1$. Both of the matrices $\rho^{(0)}$ and $\rho^{(1)}$ have two zero eigenvalues, and two non-zero eigenvalues, which are $\frac{(1 - \sqrt{c_1^2z_1^2 + c_2^2z_2^2 + c_3^2z_3^2})}{2}$, and $\frac{(1 + \sqrt{c_1^2z_1^2 + c_2^2z_2^2 + c_3^2z_3^2})}{2}$. Let $\theta = \sqrt{c_1^2z_1^2 + c_2^2z_2^2 + c_3^2z_3^2}$. Then the Sharma-Mittal entropies of $\rho^{(0)}$, and $\rho^{(1)}$ are
	\begin{equation}
	H_{q, r}\left(\rho^{(0)}\right) = H_{q, r}\left(\rho^{(1)}\right) = \frac{1}{(1 - r)}\left[\left\{\frac{(1 + \theta)^q}{2^q} + \frac{(1 + \theta)^q}{2^q} \right\}^{\frac{1 - r}{1 - q}} - 1\right].
	\end{equation}
	Now, as $p_k = \frac{1}{2}$, we have $\sum_{k = 0}^1 p_k H_{q, r}(\rho^{(k)}) = H_{q, r}(\rho^{(k)})$, which we need to maximize. Considering, $c = \max\{|c_1|, |c_2, |c_3|\}$, we have $\theta \leq \sqrt{c^2(z_1^2 + z_2^2 + z_3^2)} = c$ which is the maximum value of $\theta$. Putting it the expression of $H_{q, r}(\rho^{(0)})$ we have
	\begin{equation}\label{minimization}
	\begin{split} 
	& \max_\theta \left(H_{q, r}\left(\rho^{(0)}\right)\right) = \max_\theta \left(H_{q, r}\left(\rho^{(1)}\right)\right) \\
	& = \frac{1}{(1 - r)}\left[\left\{\frac{(1 + c)^q}{2^q} + \frac{(1 - c)^q}{2^q} \right\}^{\frac{1 - r}{1 - q}} - 1\right].
	\end{split} 
	\end{equation}
	Adding equation (\ref{entropy_i_by_2}), (\ref{minimization}) and (\ref{sharma_mittal_for_2_qubits}) we get the result.
\end{proof}  	

\begin{corollary}\label{renyi_based_discord_in_gen}
	The R{\'e}nyi quantum discord of a two qubit quantum state given by a density matrix $\rho$ is
	$$\mathcal{D}_q^{(R)}(\rho) = 1 + \frac{1}{1 - q}\log \left[ \left(\frac{1 + c}{2}\right)^q + \left(\frac{1 - c}{2}\right)^q\right] - H^{(R)}_q(\rho),$$
	where $H^{(R)}_q(\rho)$ is given by equation (\ref{renyi_for_2_qubits}).
\end{corollary}

\begin{proof}
	We have seen that the Sharma-Mittal entropy reduces to the R{\'e}nyi entropy if $r \rightarrow 1$. From equation (\ref{Renyi_entropy}) we get $H^{(R)}_q(\frac{I}{2}) = \frac{1}{1 - q}\log(\frac{1}{2^q} + \frac{1}{2^q}) = 1$. Taking $r \rightarrow 1$, in the equation (\ref{minimization}) we get, $\max_\theta (H_{q, r}(\rho_0)) = \max_\theta (H_{q, r}(\rho_1)) = $
	\begin{equation}
	\frac{1}{1 - q}\log \left[ \left(\frac{1 + c}{2}\right)^q + \left(\frac{1 - c}{2}\right)^q\right].
	\end{equation}
	Hence, the result.
\end{proof}

\begin{corollary}\label{tsallis_based_discord_in_gen}
	The Tsallis discord of a two qubit quantum state given by a density matrix $\rho$ is 
	\begin{equation*}
	\mathcal{D}^{(T)}_q(\rho) = \frac{1}{(1 - q)}\left[\left\{\frac{(1 + c)^q}{2^q} + \frac{(1 - c)^q}{2^q} \right\} - 1\right] + \frac{1}{1 - q}(2^{1 - q} - 1) - H^{(T)}_q(\rho),
	\end{equation*}
	where $H^{(T)}_q(\rho)$ is given by equation (\ref{tsallis_for_2_qubits}).
\end{corollary}

\begin{proof}
	We know that the Sharma-Mittal entropy reduces to Tsallis entropy if $r \rightarrow q$. From equation (\ref{Tsallis_entropy}) we get $H^{(T)}_q(\frac{I}{2}) = \frac{1}{1 - q}(\frac{1}{2^q} + \frac{1}{2^q} - 1) = \frac{1}{1 - q}(2^{1 - q} - 1)$. Taking $r \rightarrow q$, in the equation \ref{minimization} we get,
	\begin{equation}
	\begin{split} 
	& \max_\theta \left(H^{(T)}_q\left(\rho^{(0)}\right)\right) = \max_\theta \left(H^{(T)}_q\left(\rho^{(1)}\right)\right) \\
	& = \frac{1}{(1 - q)}\left[\left\{\frac{(1 + c)^q}{2^q} + \frac{(1 - c)^q}{2^q} \right\} - 1\right].
	\end{split} 
	\end{equation}
	Hence, the result.
\end{proof}

In the next proof, instead of using the definition of von-Neumann entropy directly, we present the limit operations on $\mathcal{D}_{q,r}(\rho)$, explicitly. It justifies that the conventional idea of quantum discord is a limiting situation of the Sharma-Mittal quantum discord. Our result matches with the expressions derived in \cite{luo2008quantum}.

\begin{corollary}\label{vn_based_discord_in_gen}
	The von-Neumann discord is given by 
	$$\mathcal{D}(\rho) = 2-\left(\frac{1+c}{2}\right) \log\left(1+c\right) -\left(\frac{1-c}{2}\right) \log\left(1-c\right)- H^{(S)}_q(\rho),$$
	where $H^{(S)}_q(\rho)$ is the von-Neumann entropy of $\rho$.
\end{corollary}

\begin{proof}
	We know that the Sharma-Mittal entropy reduces to Shannon entropy if $r \rightarrow 1$ and $q \rightarrow 1$. Now,
	\begin{equation}
		H\left(\frac{I}{2}\right) = \lim_{r \rightarrow 1} H^{(T)}_q\left(\frac{I}{2}\right) = \lim_{r \rightarrow 1} \frac{1}{1 - r}(2^{1 - r} - 1) = \lim_{r \rightarrow 1} 2^{1-r} \log 2 = \log 2 = 1,
	\end{equation}  
	using the {\'L}H{\^o}pital's rule. Taking $r \rightarrow 1$ and $q \rightarrow 1$, in the equation (\ref{minimization}) we get,
	\begin{equation*}
	\begin{split}
	&\lim_{(q, r) \rightarrow (1, 1)} \max_\theta (H_{q, r}(\rho^{(0)})) =  \lim_{(q, r) \rightarrow (1, 1)} \max_\theta (H_{q, r}(\rho^{(1)}))\\
	& = \lim_{(q, r) \rightarrow (1, 1)}\frac{1}{(1 - r)}\left[\left\{\left(\frac{1 + c}{2}\right)^q + \left(\frac{1 - c}{2}\right)^q \right\}^{\frac{1-r}{1-q}} - 1\right] \\
	&=\lim_{r \rightarrow 1}\lim_{q \rightarrow r}\frac{1}{(1 - r)}\left[\left\{\left(\frac{1 + c}{2}\right)^q + \left(\frac{1 - c}{2}\right)^q \right\}^{\frac{1-r}{1-q}} - 1\right] \\
	&=\lim_{r \rightarrow 1}\frac{1}{(1 - r)}\left[\left\{\left(\frac{1 + c}{2}\right)^r + \left(\frac{1 - c}{2}\right)^r \right\} - 1\right].
	\end{split}
	\end{equation*}
	Using the {\'L}H{\^o}pital's rule the above expression, $\lim_{(q, r) \rightarrow (1, 1)} \max_\theta (H_{q, r}(\rho^{(0)})) =  \lim_{(q, r) \rightarrow (1, 1)} \max_\theta (H_{q, r}(\rho^{(1)}))$
	\begin{equation*}
	\begin{split}
	& = - \lim_{r \rightarrow 1} \left[\left(\frac{1+c}{2}\right)^r \log \left(\frac{1+c}{2}\right) + \left(\frac{1-c}{2}\right)^r \log \left(\frac{1-c}{2}\right)\right] \\
	& = -\left(\frac{1+c}{2}\right) \log \left(\frac{1+c}{2}\right) - \left(\frac{1-c}{2}\right) \log \left(\frac{1-c}{2}\right) \\
	& = 1-\left(\frac{1+c}{2}\right) \log \left(1+c\right) - \left(\frac{1-c}{2}\right) \log \left(1-c\right).
	\end{split}
	\end{equation*}
	Applying these limiting values in $\lim_{(q, r) \rightarrow (1,1)}\mathcal{D}_{q,r}(\rho)$, mentioned in the theorem \ref{sm_based_discord_in_gen} we get the result.
\end{proof}

Now, we discuss Sharma-Mittal quantum discord and its specifications for a number of well known two qubit quantum states. For pure states it is simplified to the entropy of a state which is equivalent to the entropy of entanglement.

\begin{corollary}
	The Sharma-Mittal quantum discord of any pure entangled state is nonzero.
\end{corollary}

\begin{proof}
	The density matrix $\rho = \ket{\psi}\bra{\psi}$ representing a pure state $\ket{\psi}$ has only one non-zero eigenvalue, which is $1$. Therefore, $H_{q,r}(\rho) = 0$. A set of measurement operators produce an ensemble of pure states from a given pure state $\ket{\psi}$. Therefore, the Sharma-Mittal entropy for all these states are also $0$, that is $\max_{\{\Pi_i\}}\left[\sum_i p_i H_{q,r}(\rho_a^{(i)})\right] = 0$. Now, the definition \ref{sm_discord_deff} suggests that $\mathcal{D}_{q,r}(\rho) = H_{q,r}(\rho_b)$. If $\ket{\psi}$ is an entangled state, $\rho_b$ is mixed state and $H_{q,r}(\rho_b) \neq 0$. Therefore, $\mathcal{D}_{q,r}(\rho) \neq 0$. 
\end{proof}

Similarly, the R{\'e}nyi and Tsallis discord for pure bipartite entangled state is nonzero and reduces to $H^{(R)}_{q}(\rho_b)$, and $H^{(T)}_{q}(\rho_b)$, respectively.

\subsection{Werner State}

A $d \times d$ dimensional Werner state \cite{werner1989quantum} is a mixture of  symmetric and antisymmetric projection operators, represented by the density matrix,
\begin{equation}
\rho = \frac{2p}{d^2 + d}P_{sym} + \frac{2(1 - p)}{d^2 - d}P_{ass},
\end{equation}
where $0 \leq p \leq 1$, $P_{sym}=\frac{1}{2}(1+P)$, and $P_{ass}=\frac{1}{2}(1-P)$ are the projectors as well as $P = \sum_{ij}\ket {i}\bra{j} \otimes \ket{j}\bra{i}$.
When $d = 2$, the density matrix representing a Werner state is given by 
\begin{equation}\label{werner_density_matrix}
\rho = \begin{bmatrix} \frac{p}{3} & 0 & 0 & 0 \\
0 & -\frac{p}{3} + \frac{1}{2} & \frac{2p}{3}-\frac{1}{2} & 0 \\ 0 & \frac{2p}{3}-\frac{1}{2} & -\frac{p}{3}+\frac{1}{2} & 0 \\ 0 & 0 & 0 & \frac{p}{3} \end{bmatrix}.
\end{equation}

\begin{theorem}
	The Sharma-Mittal quantum discord for $2$-qubit Werner state mentioned in equation (\ref{werner_density_matrix}) is
	\begin{equation*}
	\begin{split}
	\mathcal{D}_{q,r}(\rho) = \frac{2^{1 - r} - 1}{1 - r} + \frac{1}{1-r}\left[\left\{ \left(\frac{3+|4p-3|}{6}\right)^q + \left(\frac{3-|4p-3|}{6}\right)^q\right\}^{\frac{1 - r}{1 - q}}\right]& \\
	- \frac{1}{1-r}\left[\left\{(1 - p)^q + 3\left(\frac{p}{3}\right)^q\right\}^{\frac{1-r}{1-q}}\right]&.
	\end{split}
	\end{equation*}
\end{theorem}
\begin{proof}
	Comparing the density matrix in equation (\ref{werner_density_matrix}) with density matrix considered in the equation (\ref{density_matrix_expanded}) we get, 
	\begin{equation} 
	1+c_{3}=\frac{4p}{3}; -c_{2}+c_{1} = 0 ; c_{1}+c_{2}=-2+\frac{8p}{3}; -c_{3}+1=2-\frac{4p}{3}.
	\end{equation}
	They indicate $c_{1} = c_{2} = c_{3} = \frac{4p}{3}-1$. Hence, $c = \max\{|c_1|, |c_2|, |c_3|\} = |\frac{4p}{3}-1|$. Also, the eigenvalues of $\rho$ are $(1 - p)$, and $\frac{p}{3}$ with multiplicity $3$. Therefore, the Sharma-Mittal entropy for a given isotropic state $\rho$ is 
	\begin{equation}
	H_{q,r}(\rho) = \frac{1}{1-r}\left[\left\{(1 - p)^q + 3\left(\frac{p}{3}\right)^3\right\}^{\frac{1-r}{1-q}}-1 \right].
	\end{equation}
	Also, replacing $c = |\frac{4p}{3}-1|$ in equation (\ref{minimization}) the maximization term reduces to
	\begin{equation}
	\begin{split} 
	& \max_\theta (H_{q, r}(\rho^{(0)})) = \max_\theta (H_{q, r}(\rho^{(1)})) \\
	& = \frac{1}{1 - r}\left[\left\{ \left(\frac{3+|4p-3|}{6}\right)^q + \left(\frac{3-|4p-3|}{6}\right)^q\right\}^{\frac{1 - r}{1 - q}} - 1\right].
	\end{split} 
	\end{equation}
	Combining we get the expression of the Sharma-Mittal quantum discord of $2$-qubit Werner state, mentioned in the statement.
\end{proof}

The Sharma Mittal discord for Werner state with respect to the state parameter $p$ and entropy parameter $q$ is plotted in fig \ref{Sharma_Werner_par_q}. Also, the figure \ref{Sharma_Werner_par_r} shows the Sharma-Mittal discord of Werner state with respect to $p$ and $r$ keeping $p$ unchanged.
\begin{figure}
	\centering
	\includegraphics[scale = .3]{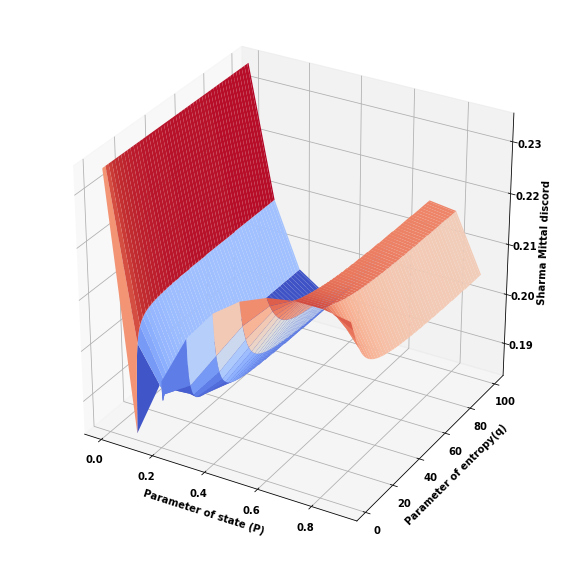}
	\caption{Sharma Mittal discord is plotted for Werner state with respect to the state parameter $p$ and entropy parameter $q$ keeping $r=5$. (Colour online)}
	\label{Sharma_Werner_par_q}
\end{figure}

\begin{figure}
	\centering
	\includegraphics[scale = .3]{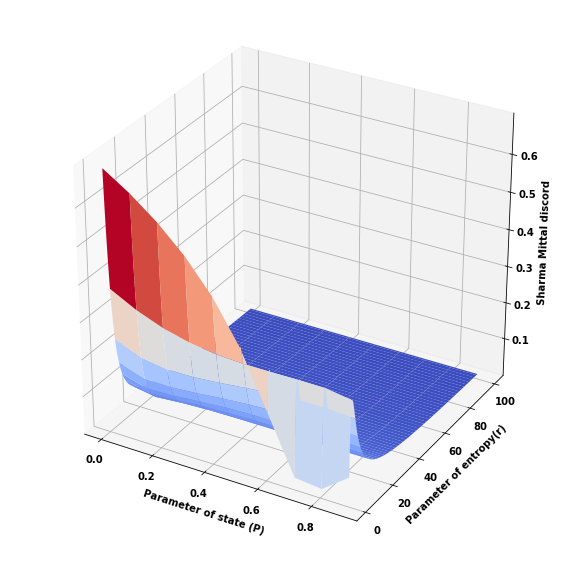}
	\caption{Sharma Mittal discord is plotted for Werner state with respect to the state parameter $p$ and entropy parameter $r$ keeping $q=5$. (Colour online)}
	\label{Sharma_Werner_par_r}
\end{figure}

\begin{corollary}
	The R{\'e}nyi discord for $2$-qubit Werner state $\rho$ is given by 
	\begin{equation}
	\begin{split}
	\mathcal{D}^{(R)}_q(\rho) = 1 - \frac{1}{1 - q} \left[\log \left\{\left(\frac{3+|4p-3|}{6}\right)^q + \left(\frac{3-|4p-3|}{6}\right)^q\right\}\right] & \\
	+ \frac{1}{1 - q} \left[\log \left\{3\left(\frac{p}{3}\right)^q + (1 - p)^q\right\}\right]&.
	\end{split}
	\end{equation}
\end{corollary}
\begin{proof}
	The R{\'e}nyi entropy of two qubit Werner state is
	\begin{equation}
	H^{(R)}_q(\rho) = -\frac{1}{1 - q} \log \left[3\left(\frac{p}{3}\right)^q + (1 - p)^q\right].
	\end{equation}
	Also, the maximization term is given by 
	\begin{equation}
	\begin{split} 
	& \max_\theta (H^{(R)}_q(\rho^{(0)})) = \max_\theta (H^{(R)}_q(\rho^{(1)}))\\
	& = -\frac{1}{1 - q}\log \left[\left\{ \left(\frac{3+|4p-3|}{6}\right)^q + \left(\frac{3-|4p-3|}{6}\right)^q\right\}\right].
	\end{split} 
	\end{equation}
	Now, combining we get the result from corollary \ref{renyi_based_discord_in_gen}.
\end{proof}
The R{\'e}nyi discord $\mathcal{D}_q^{(R)}(\rho)$ is plotted with respect to $p$ and $q$ in the figure \ref{Renyi_werner}.
\begin{figure}
	\centering
	\includegraphics[scale = .3]{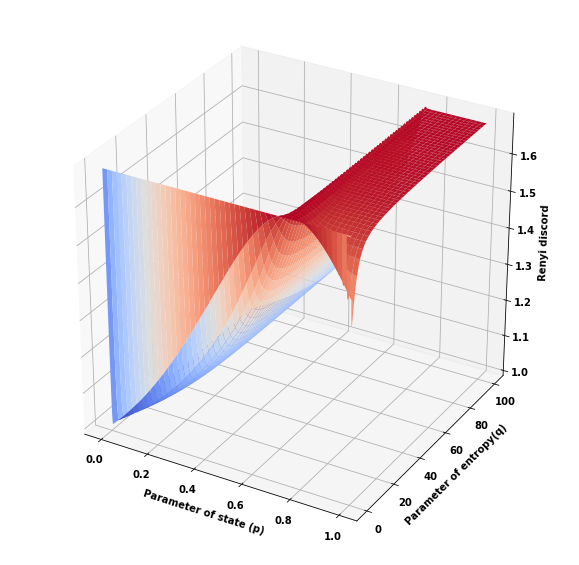}
	\caption{R{\'e}nyi discord is plotted for Werner state with respect to the state parameter $p$ and entropy parameter $q$. (Colour online)}
	\label{Renyi_werner}
\end{figure}

\begin{corollary}
	The Tsallis discord of $2$-qubit Werner state is given by
	\begin{equation*}
	\begin{split}
	\mathcal{D}^{(T)}_q(\rho) = \frac{1}{1 - q}[\left(\frac{3+|4p-3|}{6}\right)^q + \left(\frac{3-|4p-3|}{6}\right)^q + (2^{1 - q} - 1)& \\
	- 3\left(\frac{p}{3}\right)^q - (1 - p)^q ]&.
	\end{split} 
	\end{equation*} 
\end{corollary}
\begin{proof}
	The Tsallis entropy of two qubit Werner state is 
	\begin{equation}
	H^{(T)}_q(\rho) = \frac{1}{1 - q}\left[3\left(\frac{p}{3}\right)^q + (1 - p)^q - 1\right].
	\end{equation}
	Also, the maximization term reduces to 
	\begin{equation}
	\frac{1}{1 - q}\left[\left(\frac{3+|4p-3|}{6}\right)^q + \left(\frac{3-|4p-3|}{6}\right)^q - 1\right].
	\end{equation}
	Hence, the corollary \ref{tsallis_based_discord_in_gen} leads us to the proof. 
\end{proof}
The Tsallis discord of the Werner state $\rho$ with respect to $p$ and $q$ is plotted in the figure \ref{Tsallis_werner}.
\begin{figure}
	\centering
	\includegraphics[scale = .3]{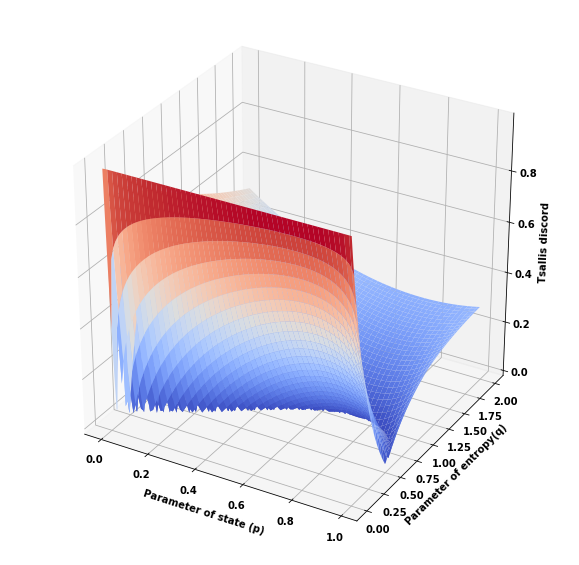}
	\caption{Tsallis discord is plotted for Werner state with respect to the state parameter $p$ and entropy parameter $q$. (Colour online)}
	\label{Tsallis_werner}
\end{figure}

\begin{corollary}
	The von-Neumann quantum discord for $2$-qubit Werner state is given by
	\begin{equation}
	\begin{split}
	\mathcal{D}(\rho) = & 1 - \left(\frac{3+|4p-3|}{6}\right) \log \left(\frac{3+|4p-3|}{6}\right)\\
	& - \left(\frac{3-|4p-3|}{6}\right) \log \left(\frac{3-|4p-3|}{6}\right)\\
	& + \left(1-p\right) \log \left(1-p\right) + p \log \frac{p}{3}.
	\end{split}
	\end{equation}
\end{corollary}
\begin{proof}
	The von-Neumann entropy of two qubit Werner state is 
	\begin{equation}
	\begin{split}
	H(\rho) & = \lim_{(q, r) \rightarrow (1, 1)} \frac{1}{1-r} \left[\left[\left(1 - p \right)^q + 3 \left( \frac{p}{3} \right)^q \right]^{\frac{1 - r}{1 - q}} - 1\right] \\
	& = \lim_{r \rightarrow 1}\lim_{q \rightarrow r} \frac{1}{1-r} \left[\left[\left(1 - p\right)^q + 3 \left(\frac{p}{3}\right)^q\right]^{\frac{1 - r}{1 - q}} -1\right] \\
	&=\lim_{r \rightarrow 1}\frac{1}{1-r}\left[\left[\left(1-p\right)^r+3\left(\frac{p}{3}\right)^r\right]-1\right].
	\end{split}
	\end{equation}
	Using the {\'L}H{\^o}spital's rule we get
	\begin{equation}
	\begin{split}
	H(\rho) & = \lim_{r \rightarrow 1}\left[ - \left(1-p\right)^r \log \left(1-p\right) - 3 \left(\frac{p}{3}\right)^r\log\left(\frac{p}{3}\right) \right] \\
	& = - \left(1-p\right) \log \left(1-p\right) - p\log\left(\frac{p}{3}\right).
	\end{split}
	\end{equation}
	In addition, the maximization term reduces to 
	\begin{equation}
	\begin{split}
	- \left(\frac{3+|4p-3|}{6}\right) \log \left(\frac{3+|4p-3|}{6}\right)- \left(\frac{3-|4p-3|}{6}\right) \log \left(\frac{3-|4p-3|}{6}\right).
	\end{split}
	\end{equation}
	Hence, the corollary \ref{vn_based_discord_in_gen} leads us to the proof.
\end{proof}

The negativity of entanglement is a well known measure of entanglement. It is the negated sum of negative eigenvalues of the partial transpose of density matrix. The partial transpose of $\rho$ with respect to the subsystem $\mathcal{H}^{(B)}$ is
\begin{equation}\label{parwerner}
\rho^\tau = \begin{bmatrix} \frac{p}{3} & 0 & 0 & \frac{2p}{3}-\frac{1}{2} \\ 0 & -\frac{p}{3} + \frac{1}{2} & 0 & 0 \\ 0 & 0 & -\frac{p}{3}+\frac{1}{2} & 0 \\ \frac{2p}{3}-\frac{1}{2} & 0 & 0 & \frac{p}{3} \end{bmatrix}.
\end{equation}
Eigenvalues of this matrix are $(p - \frac{1}{2})$, and $(\frac{1}{2} - \frac{p}{3})$ with multiplicity three. The first eigenvalue if negative when $p \leq \frac{1}{2}$. The second one is everywhere positive in $0 \leq p \leq 1$. Therefore, entanglement of $\rho$ is
\begin{equation}
N(\rho) = \frac{1}{2} \left[\left|p - \frac{1}{2}\right| - \left(p - \frac{1}{2}\right)\right].
\end{equation}
In the figure \ref{togther_werner}, we plot negativity, Sharma-Mittal, Renyi, Tsallis and Von Neumann discord for Werner state. Here, we can recognize different behavior of quantum correlations with the state parameter $p$. The quantum entanglement is zero for $ p \geq .5$, were all the discords has non-zero value. The R{\'e}nyi entropy, marked by a bulleted ($\bullet$) line takes higher value in comparison to the other discords. The Sharma-Mittal and Tsallis discords are also non-negative when $p \geq .5$ but they do not follow the entanglement like the von-Neumann discord. It indicates applicability of these quantum discords in different tasks of quantum information with zero discord quantum states. Note that, $\mathcal{D}_{.5, .4}(\rho)$ and  $\mathcal{D}^{(T)}_{.5}(\rho)$ are negative for some states.
\begin{figure}
	\centering
	\includegraphics[scale = .3]{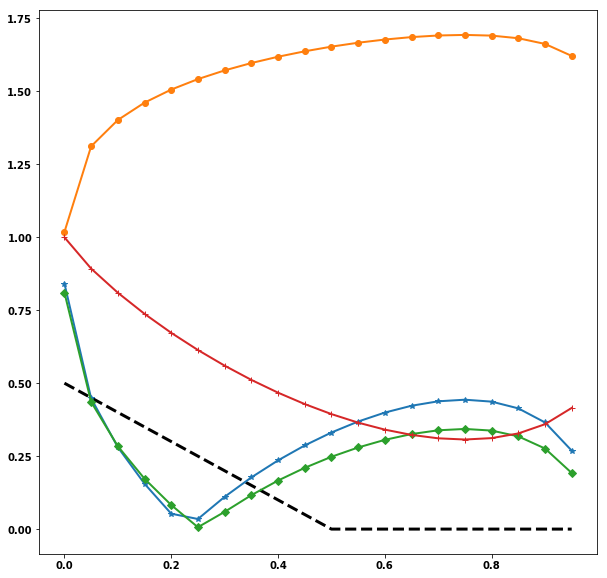}
	\caption{The broken line represents negativity of Werner state. The Sharma-Mittal discord, R{\'e}nyi,  Tsallis and Von Neumann discord are represented by lines with markers $\star$, $\bullet$,  $\blacklozenge$ and $+$ respectively. In all the plots $q = .5, r = .4$. (Colour online)}
	\label{togther_werner}
\end{figure}

\subsection{Isotropic State}

An isotropic state is a bipartite quantum state which is given by
\begin{equation}
\rho=\frac{d^2}{d^2-1}\left[(1-F)\frac{I}{d^2} + \left( F - \frac{1}{d^2}\right) \ket{\phi^+}\bra{\phi^+}\right],
\end{equation}	where $0 \leq F \leq 1$ and $\ket{\phi^+} = \frac{1}{\sqrt{d}}\sum_j \ket{j}\otimes \ket{j}$, a maximally entangled state in $\mathcal{H}^{(d)} \otimes \mathcal{H}^{(d)}$. In general, a two-qubit isotropic state is 
\begin{equation}
\rho = \begin{bmatrix}\label{denisotropic}
\frac{F}{3}+\frac{1}{6} & 0 & 0 & \frac{2}{3}F-\frac{1}{6} \\ 0 & \frac{1}{3}-\frac{F}{3} & 0 & 0 \\ 0 & 0 & \frac{1}{3}-\frac{F}{3} & 0 \\ \frac{2}{3}F-\frac{1}{6} & 0 & 0 & \frac{F}{3}+\frac{1}{6} \end{bmatrix}.
\end{equation}
The density matrices mentioned in equation(\ref{parwerner}) and equation(\ref{denisotropic}) are centro-symmetric. Entanglement of these quantum states is investigated in \cite{fel2012quantum}. Comparing equation(\ref{denisotropic}) with equation (\ref{density_matrix}) we get,
\begin{equation}
\begin{split}
& 1+c_{3} = 4\left(\frac{F}{3}+\frac{1}{6}\right); -c_{2}+c_{1}=4\left(\frac{2F}{3}-\frac{1}{6}\right);\\
& c_{1}+c_{2}=0; 1-c_{3}=4\left(\frac{1}{3}-\frac{F}{3}\right),
\end{split}
\end{equation}
which indicate
\begin{equation}
c_{1} = \frac{4F}{3}-\frac{1}{3}; c_{2} = -\frac{4F}{3} + \frac{1}{3}; c_{3} = \frac{4F}{3} - \frac{1}{3}.
\end{equation}
Hence, $c = \max\{|c_1|, |c_2|, |c_3|\} = |\frac{4F}{3} - \frac{1}{3}|$.

\begin{theorem}
	The Sharma-Mittal quantum discord for two qubit isotropic states $\rho$ is 
	\begin{equation}
	\begin{split}
	\mathcal{D}_{q,r}(\rho) = \frac{1}{1 - r}\left[\left\{ \left(\frac{3 + |4F - 1|}{6}\right)^q + \left(\frac{3 - |4F - 1|}{6}\right)^q \right\}^{\frac{1 - r}{1 - q}}\right] & \\
	+ \frac{2^{1 - r} - 1}{1 - r} - \frac{1}{1 - r}\left[\left\{F^q + 3\left(\frac{1 - F}{3}\right)^q \right\}^{\frac{1 - r}{1 - q}}\right]&.
	\end{split}
	\end{equation} 
\end{theorem}

\begin{proof}
	Note that, the eigenvalue of $\rho$ are $F$, and $\frac{(1 - F)}{3}$ with multiplicity $3$. Hence, the Sharma-Mittal quantum entropy of $\rho$ is 
	\begin{equation}
	H_{q,r}(\rho) = \frac{1}{1 - r}\left[\left\{F^q + 3\left(\frac{1 - F}{3}\right)^q \right\}^{\frac{1 - r}{1 - q}} - 1\right].
	\end{equation}
	Considering $c = |\frac{4F}{3} - \frac{1}{3}|$ the maximization terms gets the form
	\begin{equation}
	\begin{split} 
	& \max_{\theta}(H_{q,r}(\rho^{(0)})) = \max_{\theta}(H_{q,r}(\rho^{(1)})) \\
	& = \frac{1}{1 - r}\left[\left\{ \left(\frac{3 + |4F - 1|}{6}\right)^q + \left(\frac{3 - |4F - 1|}{6}\right)^q \right\}^{\frac{1 - r}{1 - q}} - 1\right].
	\end{split} 
	\end{equation}
	Now, theorem \ref{sm_based_discord_in_gen} leads up to the result.
\end{proof}

The Sharma-Mittal discord for isotropic state with respect to the state parameter $F$ and entropy parameter $q$ is plotted in the figure \ref{Sharma_isotropic_par_q}. Also, the Sharma-Mittal discord for isotropic state with respect to $F$ and $r$ is plotted in the figure \ref{Sharma_isotropic_par_r}.
\begin{figure}
	\centering
	\includegraphics[scale = .3]{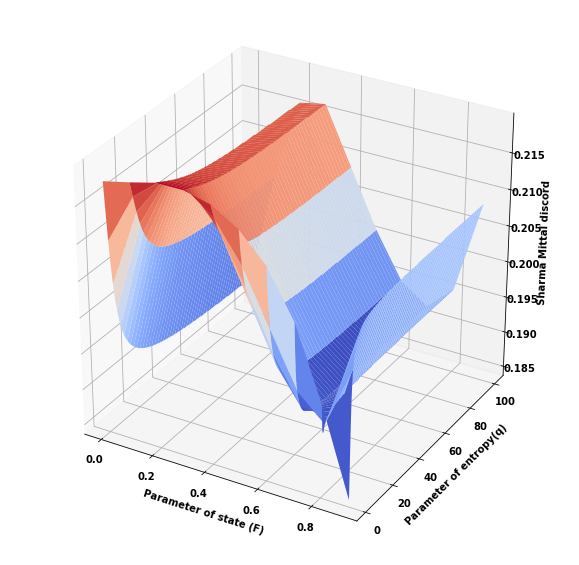}
	\caption{Sharma Mittal discord is plotted for isotropic state with respect to the state parameter $F$ and entropy parameter $q$ keeping $r=5$. (Colour online)}
	\label{Sharma_isotropic_par_q}
\end{figure}
\begin{figure}
	\centering
	\includegraphics[scale = .3]{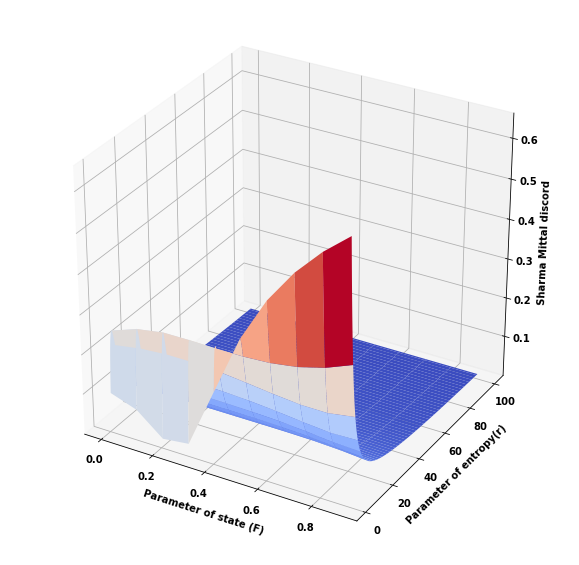}
	\caption{Sharma Mittal discord is plotted for isotropic state with respect to the state parameter $F$ and entropy parameter $r$ keeping $q=5$. (Colour online)}
	\label{Sharma_isotropic_par_r}
\end{figure}

\begin{corollary}
	The R{\'e}yi quantum discord for $2$-qubit isotropic state is  
	\begin{equation*}
	\begin{split}
	\mathcal{D}^{(R)}(\rho) =  1 - \frac{1}{1 - q} \left[\log\left\{ \left(\frac{3+|4F-1|}{6}\right)^q + \left(\frac{3-|4F-1|}{6}\right)^q\right\}\right]&\\
	+\frac{1}{1 - q} \left[ \log \left\{3\left(\frac{1-F}{3}\right)^q + F^q\right\}\right]&.
	\end{split}
	\end{equation*}
\end{corollary}

\begin{proof}
	The R{\'e}yi entropy is given by
	\begin{equation}
	H^{(R)}_q(\rho) = \frac{-1}{1-q} \log\left[F^q+3\left(\frac{1-F}{3}\right)^q\right],
	\end{equation} 
	and the maximization term reduces to 
	\begin{equation}
	\frac{-1}{1 - q}\log \left[\left\{ \left(\frac{3+|4F-1|}{6}\right)^q + \left(\frac{3-|4F-1|}{6}\right)^q\right\}\right].
	\end{equation} 
	Hence, the result follows from corollary \ref{renyi_based_discord_in_gen}.
\end{proof}

The R{\'e}nyi discord is for isotropic state plotted in the figure \ref{Renyi_isotropic}.
\begin{figure}
	\centering
	\includegraphics[scale = .3]{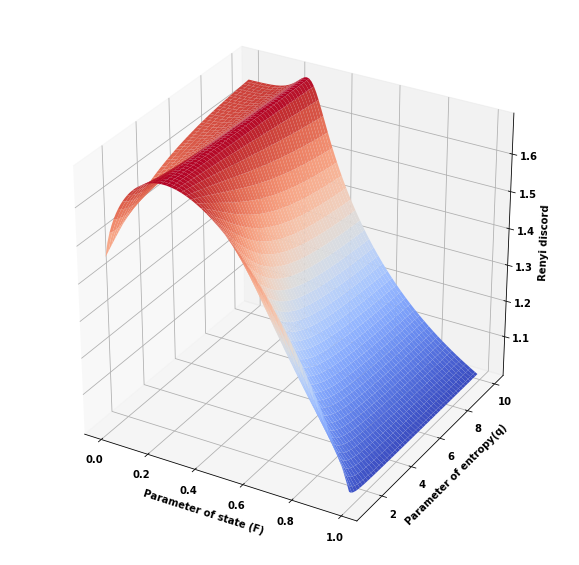}
	\caption{R{\'e}nyi discord is plotted for isotropic state with respect to the state parameter $p$ and entropy parameter $q$. (Colour online)}
	\label{Renyi_isotropic}
\end{figure}

\begin{corollary}
	The Tsallis quantum discord for two qubit isotropic state is  
	\begin{equation*}
	\begin{split}
	\mathcal{D}^{(T)}(\rho) = \frac{1}{1-q}[2^{1-q} - 1 + \left( \frac{3 + |4F - 1|}{6} \right)^q + \left( \frac{3 - |4F - 1|}{6} \right)^q& \\
	- F^q - 3\left(\frac{1 - F}{3}\right)^q].&
	\end{split}
	\end{equation*}
\end{corollary}

\begin{proof}
	The Tsallis entropy is given by
	\begin{equation}
	H^{(T)}_q(\rho) = \frac{1}{1-q}\left[F^q+3\left(\frac{1-F}{3}\right)^q-1\right],
	\end{equation} 
	and the maximization term reduces to 
	\begin{equation}
	\frac{1}{1-q}\left[\left(\frac{3+|4F-1|}{6}\right)^q+\left(\frac{3-|4F-1|}{6}\right)^q-1\right].
	\end{equation}
	Hence, the result follows from Corollary \ref{tsallis_based_discord_in_gen} .
\end{proof}
The Tsallis discord is for isotropic state plotted in the figure \ref{Tsallis_isotropic}.
\begin{figure}
	\centering
	\includegraphics[scale = .3]{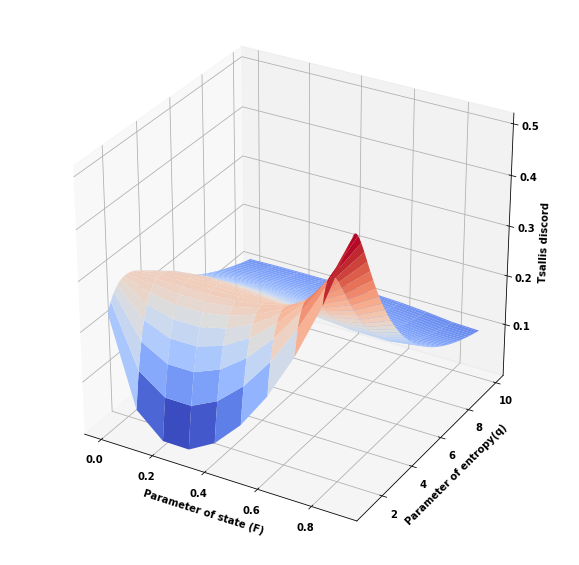}
	\caption{Tsallis discord for isotropic state is plotted with respect to the state parameter $p$ and entropy parameter $q$. (Colour online)}
	\label{Tsallis_isotropic}
\end{figure}

\begin{corollary}
	The von-Neumann quantum discord for two qubit isotropic state, given by a density matrix $\rho$, is
	\begin{equation*} 
	\begin{split}
	\mathcal{D}(\rho) = 1 - \left[\frac{3-|4F-1|}{6}\right] \log \left[\frac{3-|4F-1|}{6}\right]& \\
	-\left[\frac{3+|4F-1|}{6}\right] \log \left[\frac{3+|4F-1|}{6}\right] & \\
	+ F\log F +\left(1-F\right)\log\left(\frac{1-F}{3}\right). &
	\end{split} 
	\end{equation*}
\end{corollary}

\begin{proof}
	The von-Neumann entropy is given by 
	$$ -F\log F -\left(1-F\right)\log\left(\frac{1-F}{3}\right), $$  and the maximization term reduces to
	\begin{equation}
	\begin{split}
	-\left[\frac{3-|4F-1|}{6}\right] \log \left[\frac{3-|4F-1|}{6}\right] -\left[\frac{3+|4F-1|}{6}\right] \log \left[\frac{3+|4F-1|}{6}\right].
	\end{split}
	\end{equation}  
	Hence, the results follows from Corollary \ref{vn_based_discord_in_gen}.
\end{proof}

The eigenvalues of the partial transpose of the density matrix $\rho$ are $(\frac{1}{2} - F)$, and $(\frac{1}{6} + \frac{F}{3})$ with multiplicity $3$. The first eigenvalues is negative when $F > \frac{1}{2}$ and the second one is always positive in the range of $F$. Therefore, the negativity of entanglement for Werner state is 
\begin{equation}
N(\rho) = \frac{1}{2}\left[\left|\frac{1}{2} - F\right| - \left(\frac{1}{2} - F\right)\right].
\end{equation}

In the figure \ref{togther_isotropic}, we plot negativity, Sharma-Mittal, Renyi, Tsallis and Von Neumann discord of isotropic state. Different quantum correlations of the isotropic state also show similar characteristics as the Werner state. These graphs may vary over the entropy parameters $q$ and $r$. 
\begin{figure}
	\centering
	\includegraphics[scale = .3]{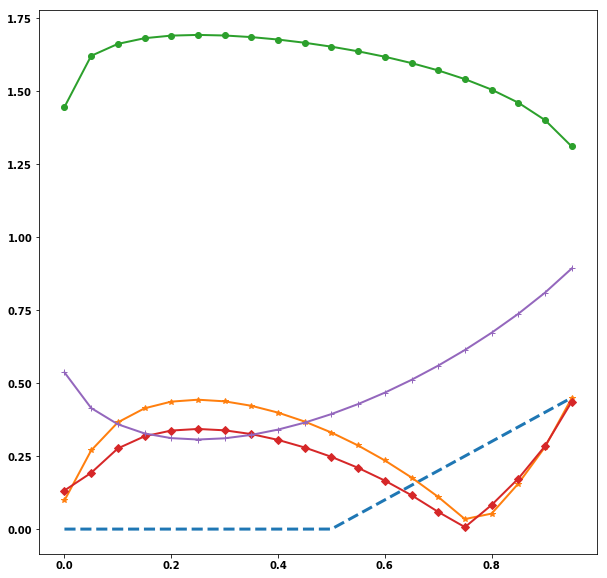}
	\caption{The broken line represents negativity of isotropic state. The Sharma-Mittal discord, R{\'e}nyi, Tsallis and Von Nemann discord are represented by lines with markers $\star$, $\bullet$,  $\blacklozenge$ and $+$, respectively. In all the plots $q = .5, r = .4$. (Colour online)}
	\label{togther_isotropic}
\end{figure}

\subsection{Pointer State}

The pointer state \cite{zurek1981pointer}, which are also called classical quantum state, is well know in the literature of quantum discord as the von-Neumann discord of pointer states are zero. It is proved that the blocks of zero discord quantum states forms a family of commuting normal matrices \cite{huang2011new}. Consider the density matrix mentioned in equation (\ref{density_matrix_expanded}). Partitioning it into four blocks we get $B_{1,1} = \frac{1}{4}\begin{bmatrix}1+c_{3} & 0  \\ 0 & 1-c_{3} \end{bmatrix}, B_{1,2} = \frac{1}{4}\begin{bmatrix}0 & c_{1}-c_{2}  \\ c_{1}+c_{2} & 0 \end{bmatrix}, B_{2, 1} = \frac{1}{4}\begin{bmatrix}0 & c_{1}+c_{2} \\ c_{1}-c_{2} & 0 \end{bmatrix}$, and $B_{2, 2} =\frac{1}{4}\begin{bmatrix} 1-c_{3} & 0  \\ 0 & 1+c_{3} \end{bmatrix}$.
Recall that a matrix $A$ is normal if $A A^\dagger = A^\dagger A$. Also, two matrices $A$ and $B$ are commutative if $AB = BA$. Therefore, all $B_{1,1}, B_{1, 2}, B_{2,1}$, and $B_{2,2}$ are normal, as well as any two of them commute. It generates a system of ten matrix equations, among which four are originated by normality and remainings are originated by commutativity conditions. Expanding these equations, one can observe that all these equations are satisfied if any two of $c_1, c_2$, and $c_3$ are zero.

\begin{theorem}
	The Sharma-Mittal quantum discord for pointer state $\rho$ is given by,
	$$\mathcal{D}(\rho) = \frac{2^{1 - r} - 1}{1 - r} \left[1 - 2^{\frac{-q(1-r)}{1-q}}\left\{\left(1+C\right)^q+\left(1-C\right)^q\right\}^{\frac{1-r}{1-q}}\right],$$
	where $-1\leq C \leq 1$ is the non-zero member of $c_1, c_2$, and $c_3$.
\end{theorem}

\begin{proof}
	If any two of $c_1, c_2$, and $c_3$ are zeros, the equation (\ref{eigenvalues_in_gen}) suggests that the eigenvalues of $\rho$ are $\frac{(1 - C)}{4}$ with multiplicity $2$, and $\frac{(1 + C)}{4}$ with multiplicity $2$. Therefore, the Sharma-Mittal entropy of $\rho$ is 
	\begin{equation}{\label{Sharma Mittal Pointer}}
	H_{q,r}(\rho) = \frac{1}{1 - r}\left[ 2^{-\frac{(2q - 1)(1 - r)}{(1 - q)}} \left\{ (1 - C)^q + (1 + C)^q \right\}^{\frac{1-r}{1 - q}} - 1\right].
	\end{equation}
	Also, as any two of $c_{1},c_{2},c_{3}$ is $0$ the maximum of them is the norzero term which is equal to $C$. Therefore, the maximization term reduces to $\max_{\theta}H_{q,r}(\rho^{(0)}) = \max_{\theta}H_{q,r}(\rho^{(1)}) = \frac{1}{1-r}\left[\left\{\left(\frac{1+C}{2}\right)^q+\left(\frac{1-C}{2}\right)^q\right\}^{\left(\frac{1-r}{1-q}\right)}-1\right]$. Then we find the result from theorem \ref{sm_based_discord_in_gen}.
\end{proof}

The Sharma-Mittal discord for pointer state with respect to the state parameter $C$ and entropy parameter $q$ is plotted in the figure \ref{Sharma_pointer_par_q}. Also, the Sharma-Mittal discord for pointer state with respect to $C$ and $r$ is plotted in the figure \ref{Sharma_pointer_par_r}. 
\begin{figure}
	\centering
	\includegraphics[scale = .3]{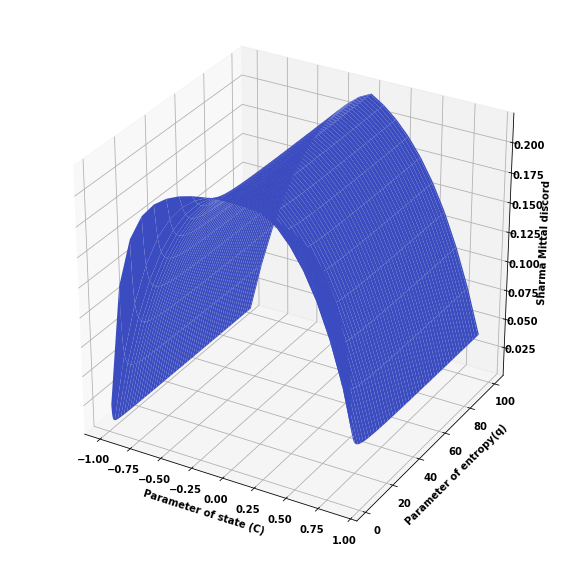}
	\caption{Sharma Mittal discord is plotted for pointer state with respect to the state parameter $C$ and entropy parameter $q$ keeping $r=5$. (Colour online)}
	\label{Sharma_pointer_par_q}
\end{figure}
\begin{figure}
	\centering
	\includegraphics[scale = .3]{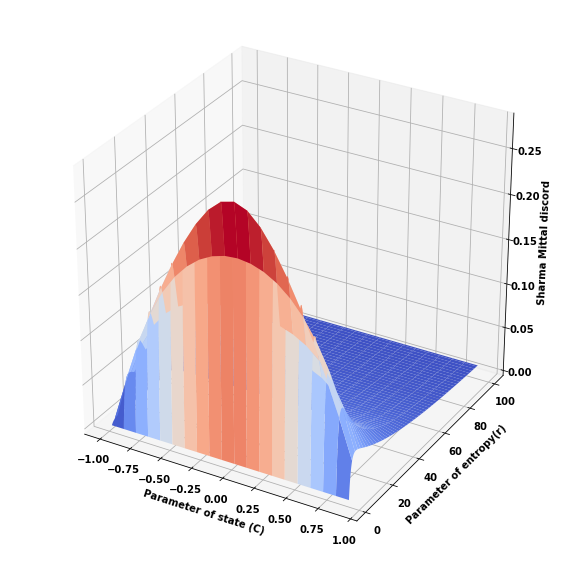}
	\caption{Sharma Mittal discord is plotted for pointer state with respect to the state parameter $C$ and entropy parameter $r$ keeping $q=5$. (Colour online)}
	\label{Sharma_pointer_par_r}
\end{figure}

\begin{corollary}
	The R{\'e}yi discord for two qubit pointer state $\rho$ is
	\begin{equation*}
	\begin{split}
	\mathcal{D}^q(\rho) = 1 - \frac{1}{1 - q} [ \log \left\{ \left( \frac{1 + C}{2} \right)^q + \left(\frac{1 - C}{2}\right)^q \right\}&\\
	-\log\left\{2^{(1 - 2q)} \left\{(1 + C)^q + (1 - C)^q \right\} \right\}]&.
	\end{split}
	\end{equation*}
\end{corollary}

\begin{proof}
	The R{\'e}nyi entropy for pointer state $\rho$ is 
	\begin{equation}
	H^{(R)}_q(\rho) = \frac{-1}{1-q}\log\left[2^{(1-2q)}\left\{(1+C)^q+(1-C)^q\right\}\right].
	\end{equation}
	Also, the maximization term reduces to $\frac{-1}{1-q}\log\left[\left(\frac{1+C}{2}\right)^q+\left(\frac{1-C}{2}\right)^q\right]$. Now from corollary \ref{renyi_based_discord_in_gen} we get the result.
\end{proof}

The R{\'e}nyi discord is for pointer state with respect to the parameter $C$ is plotted in the figure \ref{Renyi_pointer}.
\begin{figure}
	\centering
	\includegraphics[scale = .3]{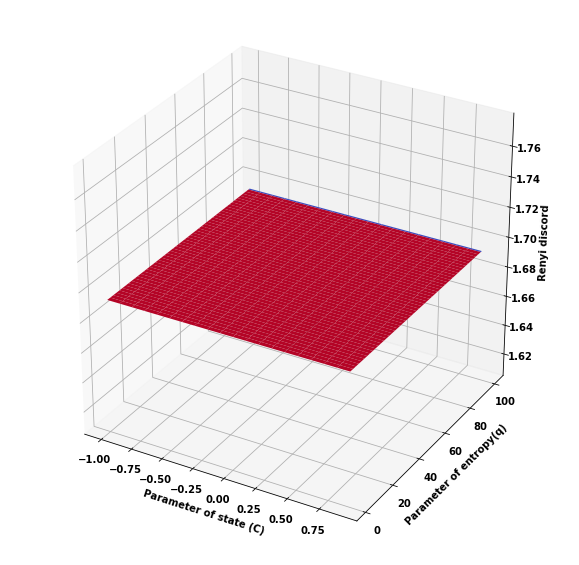}
	\caption{R{\'e}nyi discord is plotted for pointer state with respect to the state parameter $C$ and entropy parameter $q$. (Colour online)}
	\label{Renyi_pointer}
\end{figure}

\begin{corollary}
	The Tsallish quantum discord for two qubit pointer state $\rho$ is 
	$$\mathcal{D}^{T}(\rho) = \frac{2^{1 - q} - 1}{1 - q}\left[1-2^{-q}\left\{\left(1+C\right)^q+\left(1-C\right)^q\right\}\right].$$
\end{corollary}	

\begin{proof}
	The Tsallish entropy for two qubit pointer state is 
	\begin{equation}
	H^{(T)}_q(\rho) = \frac{1}{1 - q} \left[2\left(\frac{1 - C}{4}\right)^q + 2\left(\frac{1 + C}{4}\right)^q - 1\right].
	\end{equation}
	As $\max\{c_{1},c_{2},c_{3}\}=C$ the maximization term reduces to $\frac{1}{1 - q}\left[\left(\frac{1+C}{2}\right)^q+\left(\frac{1-C}{2}\right)^q- 1\right]$. Now the result follows from corollary \ref{tsallis_based_discord_in_gen}.
\end{proof}

The Tsallis discord is for pointer state plotted with respect to the parameter $C$, in the figure \ref{Tsallis_pointer}.
\begin{figure}
	\centering
	\includegraphics[scale = .3]{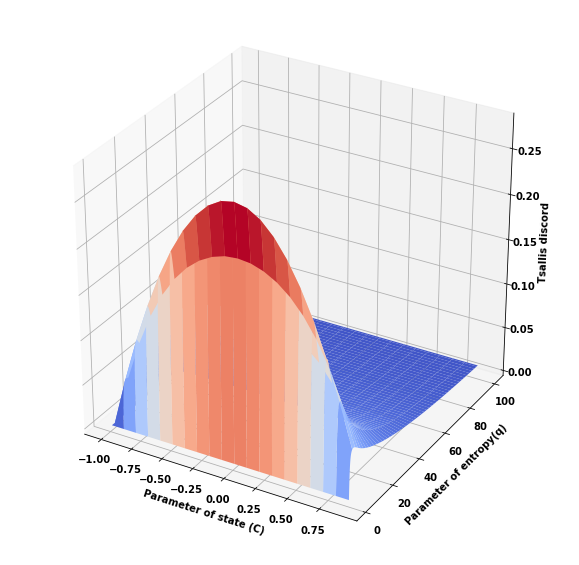}
	\caption{Tsallis discord is plotted for pointer state with respect to the state parameter $C$ and entropy parameter $q$. (Colour online)}
	\label{Tsallis_pointer}
\end{figure}

The next result is well-know in the literature of quantum discord. Here, we calculate the von-Neumann discord as a limit of Sharma-Mittal discord, and it matches with the existing idea.
\begin{corollary}
	The von-Neumann discord for $2$-qubit pointer state is $ 0 $.
\end{corollary}
\begin{proof}
	The von-Neumann entropy for two qubit pointer state is given by
	\begin{equation}
	\begin{split}
	H^{(S)}_q(\rho) = & -2\left(\frac{1+C}{4}\right) \log \left(\frac{1+C}{4}\right) - 2\left(\frac{1-C}{4}\right) \log \left(\frac{1-C}{4}\right) \\
	= & -\left(\frac{1+C}{2}\right) \log \left(\frac{1+C}{4}\right) - \left(\frac{1-C}{2}\right) \log \left(\frac{1-C}{4}\right).
	\end{split}
	\end{equation}
	The maximization term is obtained by taking limits $r \rightarrow 1, q \rightarrow 1$ in the expression for $\max_{\theta}H_{q,r}(\rho^{(0)}) = \max_{\theta}H_{q,r}(\rho^{(1)})$. Therefore, the maximization term reduces to $-\frac{1+C}{2}\log\left(\frac{1+C}{2}\right)-\frac{1-C}{2}\log\left(\frac{1-C}{2}\right)$. From the corollary \ref{vn_based_discord_in_gen} the von-Neumann discord is given by
	\begin{equation}
	\begin{split} 
	1&-\left[-\left(\frac{1+C}{2}\right) \log \left(\frac{1+C}{4}\right) - \left(\frac{1-C}{2}\right) \log \left(\frac{1-C}{4}\right)\right]\\
	&+\left[-\frac{1+C}{2}\log\left(\frac{1+C}{2}\right)-\frac{1-C}{2}\log\left(\frac{1-C}{2}\right)\right] = 0.
	\end{split}
	\end{equation}
\end{proof}

The above result suggests that the new definition of discord matches with the existing idea of discord. In the figure \ref{togther_pointer}, we plot negativity, Sharma-Mittal, Renyi, and Tsallis discord for pointer state.
\begin{figure}
	\centering
	\includegraphics[scale = .3]{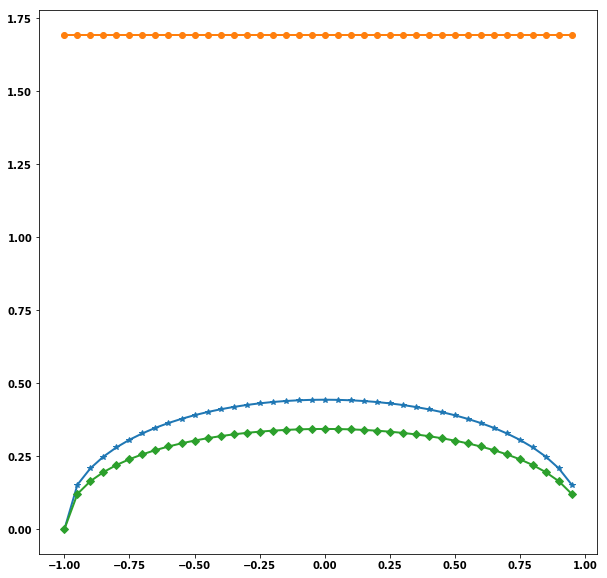}
	\caption{The Sharma-Mittal discord, R{\'e}nyi, and Tsallis discord of pointer state are represented by lines with markers $\star$, $\bullet$, and $\blacklozenge$, respectively. Here, we consider $q = .5, r = .4$. (Colour online)}
	\label{togther_pointer}
\end{figure}

\section{Conclusion and problems in future}

In this article, we have generalized the idea of quantum discord in terms of the Sharma-Mittal entropy. The R{\'e}nyi, Tsallis, and von-Neumann quantum discords are limiting cases of the Sharma-Mittal quantum discord. Analytic expressions of these discords are built up, for well-known $2$-qubit states. These new quantum correlations may be non-zero for zero entanglement states. In future, the following problems may be attempted. 

The von-Neumann quantum discord has been widely applied in quantum information theoretic tasks, for instance the remote state preparation \cite{dakic2012quantum}, device independent quantum cryptography \cite{pirandola2015high}, and many others. The new  discords may have better efficiency than the von-Neumann discord in these kind of works. For instance, the von-Neumann discord of pointer states is zero. Depending on the parameter values, the pointer states have non-zero Sharma-Mittal and other quantum discords. Hence, these new discords makes the pointer state applicable in quantum information theoretic tasks where the quantum correlation is an essential.

The pointer state has zero von-Neumann quantum discord. But, we have observed that the Sharma-Mittal, R{\'e}nyi, and Tsallis discords of the pointer state may not be zero. Given any of these definitions, there exist mixed states with zero discord. For instance, consider the Tsallis discord $\mathcal{D}^{(T)}_{.5}(\rho(p))$ of a Werner state depending on the parameter $p$. Note that, $\mathcal{D}^{(T)}_{.5}(\rho(p))$ is a continuous function of $p$. We can numerically verify that $\mathcal{D}^{(T)}_{.5}(\rho(.2546)) \times \mathcal{D}^{(T)}_{.5}(\rho(.2547)) < 0$. Therefore, there is a $p$ with $.2546 < p < .2547$ for which $\mathcal{D}^{(T)}_{.5}(\rho(p)) = 0$. Similarly, the Sharma-Mittal entropy $\mathcal{D}_{.5, .4}(\rho) = 0$, where $\rho$ is a Werner state given by a parameter $p$ such that $.2293 < p < .2294$. Now, classification of these zero discord states for different discords may also be an interesting task. 

The combinatorial aspects of von-Neumann discord is studied in \cite{dutta2017quantum}, \cite{dutta2017zero}, where a combinatorial graph represents a quantum state and the discord can be expressed in terms of graph theoretic parameters. A combinatorial study on Sharma-Mittal, R{\'e}nyi, and Tsallis discord in this direction will be very welcome.

It would be an interesting problem to formulate the Sharma-Mittal quantum discord for multipartite quantum states. In \cite{lee2017revealing} and \cite{rulli2011global} the von-Neumann discord has been investigated for tripartite and multipartite quantum states respectively. Interested reader may generalize these approaches for the Sharma-Mittal discord.

The generalization of von-Neumann quantum discord to Sharma-Mittal quantum discord introduces the idea of negative quantum correlation. This article contains examples of quantum states for which $\mathcal{D}_{q,r}(\rho) < 0$. Therefore, it demands further investigations on their applicability in the quantum information theoretic tasks, as well as a revision of the fundamental concepts of quantum correlation, such as non-negativity.

\section*{Acknowledgment}
SM and SD have equal contribution in this work. This article is similar to the one published in \textit{Quantum Information Processing 18 (6), 169}. The final authenticated version is available online at \url{https://link.springer.com/article/10.1007/s11128-019-2289-3}.


\end{document}